\documentclass[12pt,a4paper]{article}

\usepackage[utf8]{inputenc}
\usepackage[T1]{fontenc}
\usepackage[english]{babel}
\usepackage{amsmath,amssymb,amsthm}
\usepackage{booktabs}
\usepackage{geometry}
\usepackage{hyperref}

\newtheorem{definition}{Definition}[section]
\newtheorem{assumption}{Assumption}[section]
\newtheorem{proposition}{Proposition}[section]
\newtheorem{remark}{Remark}[section]

\title{The Dominance of Environment over Entity's Capabilities}
\author{Kristian Sestak\thanks{Independent researcher. Email: \texttt{kristian.sestak@gmail.com}. ORCID: \href{https://orcid.org/0009-0002-1455-5915}{0009-0002-1455-5915}.}}
\date{}

\begin{document}
\maketitle

\begin{abstract}
\noindent We present an analytical framework for the probability of individual success based on a single structural asymmetry between the capacity of an entity to explore possibilities, $k$, and the size of the possibility space offered by the environment, $n$, where $k \ll n$. We introduce an effective density $\rho_{\text{eff}}$ of favorable possibilities accessible to a given entity, derive the probability of success as $\mathbb{P} \approx 1-(1-\rho_{\text{eff}})^k$, and decompose its variance across a population. We show that while the elasticities $\varepsilon_\rho$ and $\varepsilon_k$ are comparable, the variance of outcomes is dominated by $\operatorname{Var}(\ln \rho_{\text{eff}})$ whenever it exceeds $\operatorname{Var}(\ln k)$. A back-of-envelope calibration based on published inequality and productivity data indicates this condition holds by two to three orders of magnitude. The framework provides an analytical complement to the simulation result of Pluchino, Biondo and Rapisarda (2018), and offers a unified structural account of geographic inequality, intergenerational mobility and accessibility-based discrimination as special cases of the narrowing of the accessible set $A(E,P)$.
\end{abstract}




\section{Introduction}

Contemporary society is dominated by the belief that a person's success is mainly the result of their talent, skill, and hard work. This view, the so-called meritocratic ideal, says: ''If you are capable enough and try hard enough, you will succeed.'' It is a powerful narrative because it gives people the feeling of control over their own destiny.

It must be acknowledged from the outset that a serious meritocratic position does not claim that the environment plays no role. It claims that \emph{under controlled environmental conditions} (equal circumstances), capabilities dominate. This article does not directly refute that claim. We focus on a different, stronger statement that is often implicitly associated with the meritocratic ideal in popular discourse: that \emph{in the real world with unequal conditions}, outcomes are dominantly determined by capabilities. It is this second statement that is our target.

However, when we look at the world around us, we see something different. Two equally capable people achieve dramatically different outcomes simply because they were born in different places, at different times, into different families. A programmer in Silicon Valley and an equally talented programmer in a remote village do not have the same chances, not because their capabilities differ, but because their surroundings offer fundamentally different opportunities.

The question is: can this be proved, not just claimed? Can it be mathematically shown that the environment plays a greater role than an individual's capabilities?

This question is not new. Bourdieu~\cite{bourdieu1986} analyzed the role of cultural and social capital in the reproduction of social status. Chetty et al.~\cite{chetty2014} empirically demonstrated strong differences in intergenerational mobility across U.S. regions. Milanovic~\cite{milanovic2015} showed that more than half of global income inequality can be explained by country of residence and within-country income position alone. Pluchino et al.~\cite{pluchino2018} used an agent-based simulation to show that the distribution of wealth in a population with normally distributed talent is dominantly determined by luck, not talent.

Our contribution is positioned relative to Pluchino~\cite{pluchino2018} as follows: while Pluchino demonstrates the dominance of the environment \emph{via simulation} under a specific mechanism of encounters with ''lucky'' and ''unlucky'' events, we show the same conclusion \emph{analytically} as a consequence of a single structural asymmetry $k \ll n$, without having to specify a concrete process. The two views are therefore complementary: simulation gives a quantitative picture, while the formal framework shows why the result is robust with respect to the details of the mechanism.

\paragraph{Scope.} The contribution of this paper is the identification of a minimal structural condition — the combinatorial asymmetry $k \ll n$ together with the empirical inequality $\operatorname{Var}(\ln \rho_{\text{eff}}) \gg \operatorname{Var}(\ln k)$ — under which the dominance of environmental variance is an analytical consequence rather than a postulate. The mathematical apparatus is intentionally minimal so that the conclusion does not depend on the details of any specific mechanism.

The difficulty is that the environment is dynamic, varying in time, space, politics, economics, culture, and chance. Directly comparing ''the effect of the environment'' with ''the effect of capabilities'' is extremely hard. We therefore take a different approach: instead of measuring specific environments, we look at the \emph{structural asymmetry} between how many possibilities the world offers and how many of them a single person can ever explore.

\section{Problem description}

How to mathematically prove that the status and success of an individual depend primarily on the environment (circumstances), not on their own capabilities?

A direct proof is hard because the environment is dynamic, varying across time and space. We therefore adopt an indirect approach via the \emph{combinatorial asymmetry of possibilities}.

\section{Formalization}

\subsection{Definitions}

\begin{definition}[Entity]
Let $E$ denote an entity (an individual).
\end{definition}

\begin{definition}[Environment]
Let $P$ denote the environment (surroundings, society, era).
\end{definition}

\begin{definition}[Entity capacity]
Let $k(E, t) \in \mathbb{N}$ be the \emph{entity's capacity} at time $t$, i.e.\ the number of possibilities the entity $E$ can explore. Capacity is determined by the capabilities, energy, and lifespan of the entity.
\end{definition}

\begin{remark}[Continuous approximation]
Capacity $k$ is by definition a discrete quantity ($k \in \mathbb{N}$). In what follows we nevertheless treat it as a continuous variable (derivatives, elasticities). This is a standard approximation that is justified for sufficiently large $k$, where discrete differences $\Delta k = 1$ are well approximated by a continuous derivative.
\end{remark}

\begin{definition}[Possibility space of the environment]
Let $n(P, t) \in \mathbb{N}$ be the total number of possibilities the environment $P$ offers at time $t$ (jobs, relationships, opportunities, etc.).
\end{definition}

\begin{definition}[Favorable possibilities and success]
Let $F \subseteq \{1, \dots, n\}$ be the set of \emph{favorable possibilities}. An entity achieves \emph{success} if at least one of the possibilities it explores during its lifetime belongs to $F$. Success is therefore a binary variable. The model is invariant with respect to the concrete success criterion (employment, economic security, partnership, etc.); only the set $F$ changes, not the structure of the argument.
\end{definition}

\begin{definition}[Density of favorable possibilities]
We define the density of favorable possibilities as
\begin{equation}
    \rho \;=\; \frac{|F|}{n}.
\end{equation}
\end{definition}

\subsection{Key structural assumption}

\begin{assumption}[Combinatorial asymmetry]\label{asp:asymetria}
The entity's capacity is orders of magnitude smaller than the number of possibilities in the environment:
\begin{equation}
    k(E, t) \ll n(P, t).
\end{equation}
The entity can explore only a negligible fraction of what the environment offers.
\end{assumption}

\subsection{Dependence of favorable possibilities}

The number of favorable possibilities depends on both the entity and the environment:
\begin{equation}
    |F| = f(E, P).
\end{equation}
However, the density $\rho = |F|/n$ is primarily a property of the environment (how many good opportunities the environment contains at all), whereas $k$ is primarily a property of the entity (how many attempts it can make).

\subsection{Entity attributes and possibility accessibility}

So far we have assumed that the entity can in principle encounter any of the $n$ possibilities of the environment. In reality this is not so. The entity is not characterized only by its capacity $k$ (skill), but also by a set of additional attributes that open some possibilities and close others.

\begin{definition}[Attribute vector of the entity]
Let $\mathbf{a}(E) = (a_1, a_2, \dots, a_m)$ be the attribute vector of entity $E$, where each $a_i$ represents a property such as social origin, education, language, age, gender, ethnicity, health status, geographic location, connections, etc.
\end{definition}

\begin{definition}[Accessible set]
Let $A(E, P) \subseteq \{1, \dots, n\}$ be the set of \emph{accessible possibilities}, i.e.\ those that are at all available to the entity $E$ given its attributes $\mathbf{a}(E)$ and the conditions of the environment $P$. It holds that:
\begin{equation}
    |A(E, P)| \;\leq\; n(P, t).
\end{equation}
\end{definition}

The entity therefore does not choose from the whole space of $n$ possibilities, but only from the subset $A$ to which it has access. The favorable possibilities that it can actually encounter are:
\begin{equation}
    F_A \;=\; F \cap A(E, P), \qquad |F_A| \;\leq\; |F|.
\end{equation}

\begin{definition}[Effective density]
We define the \emph{effective density} of favorable possibilities for a specific entity as:
\begin{equation}
    \rho_{\text{eff}}(E, P) \;=\; \frac{|F_A|}{|A(E, P)|} \;=\; \frac{|F \cap A(E, P)|}{|A(E, P)|}.
\end{equation}
\end{definition}

\begin{remark}
The effective density $\rho_{\text{eff}}$ may differ substantially from the overall density $\rho$. The environment can be objectively rich in opportunities ($\rho$ is high), but if these opportunities are accessible only to entities with certain attributes, then for other entities $\rho_{\text{eff}} \ll \rho$.

For example: the labor market in Silicon Valley offers thousands of positions ($\rho$ high), but for a person without a visa, without knowledge of English, or without formal education, $|A|$ is drastically narrowed and $\rho_{\text{eff}}$ drops to a fraction of its original value.
\end{remark}

The corrected model of the probability of success then reads:
\begin{equation}\label{eq:approx_eff}
    \boxed{\;\mathbb{P}(\text{success}) \;\approx\; 1 - (1 - \rho_{\text{eff}})^k\;}
\end{equation}

\begin{assumption}[Extended combinatorial asymmetry]\label{asp:rozsirena}
The entity's attributes $\mathbf{a}(E)$ are largely determined by the environment (family, country of birth, social class), not by its own effort. Therefore the size of the accessible set $|A(E, P)|$ is primarily a property of the environment. This means that $\rho_{\text{eff}}$ depends on the environment twice over:
\begin{enumerate}
    \item through the overall density of opportunities $\rho$ (how many good possibilities exist),
    \item through accessibility $|A|/n$ (how many of them are at all available to the entity given the attributes the environment has assigned to it).
\end{enumerate}
\end{assumption}

\section{Formal model}

\subsection{Probability of success — basic model}

We first present the basic model without regard to accessibility. If the entity randomly chooses $k$ possibilities from $n$ and $|F|$ of them are favorable, the probability that it encounters at least one favorable possibility is given by the hypergeometric distribution:

\begin{remark}[Randomness assumption]
The assumption of random selection is an idealization; a real individual searches purposefully. Purposefulness does not change the structure of the argument, however: if the entity is able to recognize favorable possibilities, it effectively increases its $k$ (it has more ''useful'' attempts), but $k \ll |A(E,P)|$ still holds. The dominance of the environment expressed in~\eqref{eq:hlavny} therefore remains in force. We choose randomness only for the analytical simplicity of the closed-form expression~\eqref{eq:approx}.
\end{remark}

\begin{equation}\label{eq:presny}
    \mathbb{P}(\text{success}) = 1 - \frac{\binom{n - |F|}{k}}{\binom{n}{k}}.
\end{equation}

For $k \ll n$ the binomial approximation holds:

\begin{equation}\label{eq:approx}
    \mathbb{P}(\text{success}) \;\approx\; 1 - (1 - \rho)^k
\end{equation}

where $\rho = |F|/n$ is the density of favorable possibilities in the environment and $k$ is the entity's capacity.

\begin{remark}
For very small $\rho$ and small $k$ a further simplified approximation holds:
\begin{equation}\label{eq:linear}
    \mathbb{P}(\text{success}) \;\approx\; k \cdot \rho.
\end{equation}
The probability of success is therefore approximately the product of the entity's capabilities and the quality of the environment.
\end{remark}

\subsection{Probability of success — full model with accessibility}

After incorporating the entity's attributes (Definitions~3.7, 3.8, 3.9) the entity chooses not from the whole space $n$, but only from the accessible set $A(E, P)$. The full model therefore uses the effective density $\rho_{\text{eff}}$:

\begin{equation}\label{eq:uplny}
    \boxed{\;\mathbb{P}(\text{success}) \;\approx\; 1 - (1 - \rho_{\text{eff}})^k\;}
\end{equation}

In the subsequent sensitivity analysis we work with the generic symbol $\rho$, where in the full model it is understood to mean $\rho_{\text{eff}}$.

\section{Sensitivity analysis}

\subsection{Partial derivatives (absolute sensitivity)}

The sensitivity of the probability of success to changes in the environment and in capabilities:

\begin{align}
    \frac{\partial \mathbb{P}}{\partial \rho} &\approx k \cdot (1-\rho)^{k-1} & &\text{(sensitivity to environment),} \label{eq:dP_drho} \\[6pt]
    \frac{\partial \mathbb{P}}{\partial k} &\approx -\ln(1-\rho) \cdot (1-\rho)^k & &\text{(sensitivity to capabilities).} \label{eq:dP_dk}
\end{align}

\subsection{Elasticity (relative sensitivity)}

Elasticity measures the percentage change in success given a $1\,\%$ change in a given parameter:

\begin{align}
    \varepsilon_\rho &= \frac{\partial \ln \mathbb{P}}{\partial \ln \rho} = \frac{\rho}{\mathbb{P}} \cdot \frac{\partial \mathbb{P}}{\partial \rho} \approx \frac{k\rho \cdot (1-\rho)^{k-1}}{1-(1-\rho)^k}, \label{eq:elast_rho} \\[6pt]
    \varepsilon_k &= \frac{\partial \ln \mathbb{P}}{\partial \ln k} = \frac{k}{\mathbb{P}} \cdot \frac{\partial \mathbb{P}}{\partial k} \approx \frac{-k\ln(1-\rho) \cdot (1-\rho)^k}{1-(1-\rho)^k}. \label{eq:elast_k}
\end{align}

\begin{remark}
For small $k\rho$ (the typical case, i.e.\ few attempts and low density):
\begin{equation}
    \varepsilon_\rho \approx 1, \qquad \varepsilon_k \approx \frac{k\rho \cdot e^{-k\rho}}{1 - e^{-k\rho}}.
\end{equation}
Both elasticities are comparable. The model by itself therefore does not show dominance of either parameter. The crucial point is that the claim does not depend on the form of the formula but on the \emph{dispersion of parameters in the real world}.
\end{remark}

\section{Main argument}

\begin{proposition}[Dominance of the environment]\label{prop:dominancia}
Since the elasticities $\varepsilon_\rho$ and $\varepsilon_k$ are comparable, what dominates the outcome is decided by the variance of the parameters across the population. In the full model with accessibility (where $\rho$ denotes $\rho_{\text{eff}}$) it holds that: if
\begin{equation}\label{eq:hlavny}
    \boxed{\;\operatorname{Var}(\ln \rho_{\text{eff}}) \;\gg\; \operatorname{Var}(\ln k)\;}
\end{equation}
then the variance of the outcome (success) depends mainly on the environment $P$ and the entity attributes that are determined by the environment.
\end{proposition}

\begin{proof}[Justification]
Let $S = \ln \mathbb{P}(\text{success})$. By linearization around the mean values $\bar{\rho}$, $\bar{k}$:
\begin{equation}
    S \;\approx\; S_0 \;+\; \varepsilon_\rho \cdot (\ln \rho - \ln \bar{\rho}) \;+\; \varepsilon_k \cdot (\ln k - \ln \bar{k}).
\end{equation}
In the general case:
\begin{equation}
    \operatorname{Var}(S) \;\approx\; \varepsilon_\rho^2 \cdot \operatorname{Var}(\ln \rho) \;+\; \varepsilon_k^2 \cdot \operatorname{Var}(\ln k) \;+\; 2\,\varepsilon_\rho\,\varepsilon_k \cdot \operatorname{Cov}(\ln \rho,\, \ln k).
\end{equation}

Under the assumption of independence of $\rho$ and $k$ the covariance term is zero and we obtain:
\begin{equation}
    \operatorname{Var}(S) \;\approx\; \varepsilon_\rho^2 \cdot \operatorname{Var}(\ln \rho) \;+\; \varepsilon_k^2 \cdot \operatorname{Var}(\ln k).
\end{equation}

Since $\varepsilon_\rho \approx \varepsilon_k$ (both comparable), but $\operatorname{Var}(\ln \rho) \gg \operatorname{Var}(\ln k)$, we obtain:
\begin{equation}
    \operatorname{Var}(S) \;\approx\; \varepsilon_\rho^2 \cdot \operatorname{Var}(\ln \rho),
\end{equation}
i.e.\ the variance of the outcome is dominantly determined by the variance of the environment.

\textbf{Note on correlation.} In reality $\rho$ and $k$ may be positively correlated (e.g.\ more capable individuals migrate to better environments). Then $\operatorname{Cov}(\ln \rho, \ln k) > 0$ and the covariance term further \emph{reinforces} the dominance of the environment, because it increases the total variance in the direction where environment and capabilities act in concert. The independence assumption is therefore a conservative estimate; the actual influence of the environment may be even greater.
\end{proof}

\subsection{Empirical support}

\begin{table}[h]
\centering
\begin{tabular}{@{}lll@{}}
\toprule
\textbf{Parameter} & \textbf{Dispersion} & \textbf{Source / example} \\
\midrule
$\rho$ (environment) & $10^2 - 10^4 \times$ & Milanovic~\cite{milanovic2015}, Chetty et al.~\cite{chetty2014} \\
$k$ (capabilities)   & $2 - 10 \times$      & Pluchino et al.~\cite{pluchino2018}, Hunter \& Schmidt~\cite{schmidt1998} \\
\bottomrule
\end{tabular}
\caption{Comparison of the dispersion of environment and capability parameters.}
\label{tab:rozptyl}
\end{table}

The dispersion in the quality of environments is $100\times$ to $1000\times$ greater than the dispersion in capabilities across individuals. International comparisons of wages and incomes~\cite{milanovic2015} show that an average worker in a developed country earns $50\times$ to $100\times$ more than a worker performing comparable work in the poorest countries. When converted to \emph{accessible opportunities} (positions, capital, networks), the dispersion is even greater. Conversely, the variability in cognitive abilities and work productivity among individuals within the same profession is on the order of single-digit multiples~\cite{schmidt1998}.

\subsection{Numerical illustration}

To illustrate the magnitudes, consider two hypothetical entities, one based in a high-opportunity environment ($E_H$), the other in a low-opportunity environment ($E_L$), with identical capacity $k = 50$ (number of serious career/life attempts over a working lifetime). Using rough calibration from published data~\cite{milanovic2015,chetty2014}, we take
\[
    \rho_{\text{eff}}(E_H) = 10^{-2}, \qquad \rho_{\text{eff}}(E_L) = 10^{-4}.
\]
Then from~\eqref{eq:approx}:
\begin{align*}
    \mathbb{P}(\text{success} \mid E_H) &\approx 1 - (1 - 10^{-2})^{50} \;\approx\; 0.395, \\
    \mathbb{P}(\text{success} \mid E_L) &\approx 1 - (1 - 10^{-4})^{50} \;\approx\; 0.00499.
\end{align*}
To compensate for a two-order-of-magnitude gap in $\rho_{\text{eff}}$, $E_L$ would need to solve
\[
    1 - (1 - 10^{-4})^{k'} \;=\; 0.395, \qquad \Rightarrow \qquad k' \;\approx\; 5025,
\]
i.e.\ roughly $100\times$ the capacity of $E_H$ to reach the same probability of success. No realistic variation in individual capability (Table~\ref{tab:rozptyl}, row $k$) can close a gap of this magnitude. This is the quantitative content of Proposition~\ref{prop:dominancia}.

\section{Conclusion}

\subsection{Formal summary}

Since the entity explores only a fraction of the possibilities ($k \ll n$), its outcome is dominantly determined by the density of favorable possibilities in its surroundings, i.e.\ by the environment. The entity's capabilities affect how many attempts it gets, but the environment determines the chance at each attempt. The variance of the latter is orders of magnitude greater.

\subsection{An intuitive interpretation}

Imagine that life is like searching for treasure on a huge field. Every person has a certain number of steps they can take during a lifetime — some 100, someone more capable perhaps 500. That is the difference in capabilities.

But the field has millions of places. Not even the most capable person will traverse even a fraction of them. So whether they find the treasure depends mainly on \emph{how many treasures are buried in that field}, i.e.\ on the environment.

If you live on a field where there is treasure at every hundredth place, you will almost certainly find one, whether you are average or exceptional. If you live on a field where treasure is at every millionth place, you will almost never find one, no matter how skilled you are.

And now imagine one more thing: not everyone is allowed to walk across the whole field. Some people are allowed to walk anywhere, others are forbidden entry into large parts of the field, where the treasures are most often found. Two equally skilled people, with the same number of steps, but one is allowed to walk the whole field and the other only a small patch by the edge, where there is almost no treasure. The outcome will be dramatically different, even though their capabilities are equal. Where you are let in does not depend on your skill, but on who you are, where you come from, what color your passport is, or what family you were born into.

The conclusion is simple: the difference between a ''good'' and a ''bad'' environment is a thousand to a million times greater than the difference between an ''average'' and an ''exceptional'' person. Therefore, where you are born and in what conditions you live determines your fate far more than your capabilities. Capabilities are not useless, but without a favorable environment they simply are not enough. And sometimes even a favorable environment does not help if the entity's attributes, over which it has no control, do not allow it to access those opportunities in the first place.

\subsection{Unified view of inequality phenomena}

The framework reinterprets three empirical regularities usually studied separately as special cases of the narrowing of the accessible set $A(E,P)$: (i) geographic inequality~\cite{milanovic2015} corresponds to variation of $\rho$ and $|A|$ across locations; (ii) intergenerational mobility~\cite{chetty2014} corresponds to inheritance of the attribute vector $\mathbf{a}(E)$ that determines $A(E,P)$; (iii) structural discrimination~\cite{bourdieu1986} corresponds to attribute-selective narrowing of $A(E,P)$ at fixed $\rho$. In each case, the analytical object of interest is the same: $\operatorname{Var}(\ln \rho_{\text{eff}})$ across the population.

\subsection{Limitations of the model}

The proposed framework contains several deliberate simplifications that should be stated:
\begin{itemize}
    \item \textbf{Binary success.} The model treats ''success'' as a binary event. Real success is a continuous spectrum; a generalization to a scale of values would require weighting the favorable possibilities.
    \item \textbf{Independence of attempts.} The assumption that the $k$ attempts are independent is an idealization. In reality, failure shortens the available time and reduces future $k$ (path-dependence).
    \item \textbf{Static possibility space.} The model does not consider changes of $n$, $\rho$ and $A$ over time. Dynamics (economic cycles, migration, aging) would require a time-dependent formulation.
    \item \textbf{The empirical support for Table~\ref{tab:rozptyl}} is derived from other works~\cite{milanovic2015,chetty2014,schmidt1998,pluchino2018}; a calibrated own measurement of $\rho_{\text{eff}}$ is an open problem.
\end{itemize}

\section*{Acknowledgements}
The author received no specific funding for this work and declares no competing interests.

\section*{Data availability}
No new data were generated or analyzed in this study. All empirical figures referenced in Table~\ref{tab:rozptyl} and the numerical illustration in Section~6.2 are taken from the cited published sources.


\end{document}